\documentclass[twocolumn,pra, superscriptaddress]{revtex4-1}
\usepackage{dcolumn}
\usepackage{bm}

\usepackage[dvips]{graphicx} 
\usepackage{amsfonts}
\usepackage{amssymb}
\usepackage{amscd}
\usepackage{amsmath}    
\usepackage{autobreak}
\usepackage{enumerate}
\usepackage{epsfig}
\usepackage{subfigure}
\usepackage{xcolor}
\usepackage{amsthm}

\usepackage{subfig}

\newtheorem{theorem}{Theorem}
\newtheorem{lemma}{Lemma}

\newtheorem{definition}{Definition}

\newtheorem{remark}{Remark}

\newcommand{\bra}[1]{\mbox{$\left\langle #1 \right|$}}
\newcommand{\ket}[1]{\mbox{$\left| #1 \right\rangle$}}

\newcommand{\floor}[1]{\left\lfloor #1 \right\rfloor}
\newcommand{\ceil}[1]{\left\lceil #1 \right\rceil}
\newcommand{\tens}{\mathbin{\mathop{\otimes}}}
\newcommand{\Tr}{\operatorname{Tr}}

\begin{document}

\preprint{APS/123-QED}

\title{One-shot Coherence Distillation with Catalysts }

\author{Senrui Chen}
\affiliation{Department of Electronic Engineering, Tsinghua University, Beijing 100084, China}
\affiliation{Center for Quantum Information, Institute for Interdisciplinary Information Sciences, Tsinghua University, Beijing 100084, China}
%
%
%
\author{Xingjian Zhang}
\author{You Zhou}
\affiliation{Center for Quantum Information, Institute for Interdisciplinary Information Sciences, Tsinghua University, Beijing 100084, China}
\author{Qi Zhao}
\email{zhaoqithu10@gmail.com}
\affiliation{Center for Quantum Information, Institute for Interdisciplinary Information Sciences, Tsinghua University, Beijing 100084, China}
\affiliation{Shanghai Branch, National Laboratory for Physical Sciences at Microscale and Department of Modern Physics, University of Science and Technology of China, Shanghai 201315, China}
\affiliation{CAS Center for Excellence in Quantum Information and Quantum Physics, University of Science and Technology of China, Shanghai 201315, China}

%

\date{\today}

\begin{abstract}
The resource theory of quantum coherence is an important topic in quantum information science. 
Standard coherence distillation and dilution problems have been thoroughly studied. In this paper, we introduce and study the problem of one-shot coherence distillation with catalysts. In order to distill more coherence from a state of interest, a catalytic system can be involved and a jointly free operation is applied on both systems. The joint output state should be a maximally coherent state in tensor product with the unchanged catalysts, with some allowable fidelity error. We consider several different definitions of this problem. Firstly, with a small fidelity error in both systems, we show that, even via the smallest free operation class (PIO), the distillable coherence of any state with no restriction on the catalysts is infinite, which is a ``coherence embezzling phenomenon''.
We then define and calculate a lower bound for the distillable coherence when the dimension of catalysts is restricted.
Finally, in consideration of physical relevance, we define the ``perfect catalysts'' scenario where the catalysts are required to be pure and precisely unchanged. Interestingly, we show that in this setting catalysts basically provide 
no advantages in pure state distillation via IO and SIO under certain smoothing restriction. Our work enhances the understanding of catalytic effect in quantum resource theory.
\end{abstract}

\pacs{Valid PACS appear here}
\maketitle


\section{Introduction}\label{intro}
Quantum coherence, a notion describing the superposition phenomenon, distinguishes a quantum system from its classical counterpart. Quantum coherence plays an important role in thermodynamics \cite{Aberg14,EmbezThermo,Horodecki15,Lostaglio15,lostaglio2015description,narasimhachar2015low}, quantum randomness generation \cite{cohrand,PhysRevA.97.012302,Yuan17uncertainty,ma2017source,Zhao19one}, quantum cryptography \cite{ma2018operational} and other quantum information processing tasks \cite{cohQKD,cohchem}. After several pioneering works \cite{Aaberg2006, BraunGeorgeot}, a resource framework of quantum coherence is introduced by Baumgratz et al. \cite{baumgratz14} and has greatly enhanced our understanding of this phenomenon. A typical quantum resource theory shares three basic components: a set of free states that contain no resource, a set of free operations that map any free state to a free state (hence generating no resource), and a metric functional that characterizes the amount of resource in a quantum state. One may refer to \cite{Eric18review} for a recent review on quantum resource theory. As for the resource theory of quantum coherence, the set of free states is defined as all density operators which are diagonal in a predetermined basis (incoherent basis) $\{\ket{i}\}$. There are different definitions of free operations such as maximal incoherent operations~\cite{Aaberg2006} (MIO), incoherent operations~\cite{baumgratz14} (IO), dephasing-covariant incoherent operations~\cite{Chitambar16prl, Marvian16} (DIO), strictly incoherent operations~\cite{winter16,Yadin16} (SIO), and physically implementable incoherent operations~\cite{Chitambar16prl} (PIO), showing differences in their operational ability and physical relevance. Furthermore, different measures of coherence of a state are defined, such as the $l_1$-norm, the relative entropy of coherence \cite{baumgratz14}, the coherence of formation \cite{Yuan15intrinsic, Aaberg2006,liu2018superadditivity}, the robustness of coherence \cite{Napoli16}, polynomial measure of coherence \cite{zhou2017polynomial}, etc. The operational meanings of some measures are further uncovered \cite{Yuan15intrinsic, winter16, Rana17, Liu18}. One could refer to~\cite{Streltsov-review, Hu2018review} for reviews of recent developments on the resource theory of quantum coherence. 

In a resource theory, the most fundamental operational tasks are distillation and dilution of the resource via free operations \cite{rains2001semidefinite,hayden2001asymptotic,winter16}, which exhibit the ability to manipulate the underlying resource. 
In the resource theory of coherence, similar to the well-known resource theory of entanglement,
coherence distillation is the task that transforms a given quantum state to as many as possible maximally coherent qubits $\ket{\Psi_{\mathrm{2}}} = \frac{\ket{0}+\ket{1}}{\sqrt 2}$ (in the basis $\{\ket{0},\ket{1}\}$), the unit of coherence, or a ``{cosbit}".
Dilution is the reverse process of distillation, which aims to prepare a target state with as few as possible cobits. The conversion rates of both tasks have been solved in the asymptotic cases, under the i.i.d. assumption, where infinite and identically distributed copies of the same state are available \cite{OpCoh, Lami19}.
Recently, these tasks are generalized to the one-shot scenario where only one-copy of the state is supplied \cite{zhao2018one,Regula-one-shot,Zhao19one}. {In the one-shot coherence distillation, a different target can be taken, that is, to maximize the dimension $d$ of the maximally coherent state that can be transformed from just one copy of a given state, allowing for a predetermined error $\varepsilon$. }The one-shot coherence manipulation reflects realistic experimental requests where the number of available states are finite,
and the results could be applied to quantum information processing, such as randomness generation, cryptography, and thermodynamics. 

Compared with the dilution problem, coherence distillation has more significance in  practical applications since in most of cases, we only have the access to the ``not so good'' initial states. Therefore it is quite important to investigate the limit of coherence distillations in order to boost the distillation rate. Apart from the standard distillation setting, other variant tasks have been proposed and investigated, in order to enhance the distillation process. For example, assisted coherence distillation protocols \cite{vijayan2018one,regula2018non-asymptotic,PhysRevLett.116.070402,PhysRevX.7.011024} consider distilling coherence from a system of interest with the help of another system. Any local operation on the assistant system and classical communication are allowed, while the system of interest is restricted to free operations. In this scenario, the assistant system is not required to be invariant after the procedure.

Different from the assisted distillation, an alternative way to help distill coherence is to use a catalytic system, which should remain invariant (precisely or approximately) after the distillation process. Catalysis is a concept that emerges from chemistry. A catalyst initiates a certain reaction that cannot happen without it, while itself does not change or being consumed after the process. Similar effect also exists in quantum resource theory. For some states $\rho$ and $\sigma$, there may not be free operations that transform $\rho$ to $\sigma$, but with the catalyst $\gamma$, a free operation that transforms $\rho \tens \gamma$ to $\sigma \tens \gamma$ may exist. Resource catalyst is first discovered in entanglement theory \cite{Jonathan99}, and important results of precise catalytic transformation of pure entanglement have been reached in \cite{Turgut07, Klimesh07}. The idea of catalysts is also generalized to thermodynamics \cite{EmbezThermo,catacohpossible} and coherence \cite{Bu16}. 
Since resource catalysts are shown to be useful in various tasks, it is reasonable to expect that they can enhance the distillation process.

In this paper, we introduce and systematically study the problem of catalytic distillation of coherence.
We first give the definition of catalytic coherence distillation with unrestricted dimension of the catalysts and allowing global smoothing parameter. Smoothing is a well-used technique in the one-shot convertibility between quantum states, which allows a small error in fidelity between the target state and the realistic output state, due to both mathematical and physical considerations. A global smoothing means a fidelity error involving both the original system and the catalytic system, see for instance \cite{CataDec, CataRes}. 
Surprisingly, we find that under this definition one can distill an infinite amount of coherence from any initial state with arbitrarily small error, which we call a ``coherence embezzling phenomenon''. Two different embezzling protocols are proposed to illustrate this phenomenon. We further explore the distillable coherence with catalysts of a restricted dimension, and calculate a lower bound of this quantity with the protocols discussed before. Next, in consideration of physical relevance, we propose another definition, the distillable coherence with perfect catalysts, where the catalysts are required to be pure and precisely unchanged, and smoothing is only allowed in the original distilled system. We show that, in this scenario catalysts basically provide no advantages in pure state distillation via IO and SIO, under a ``pure state smoothing'' restriction which will be explained later in this paper, hence partly solve this problem. 
\begin{figure}
    \centering
    \includegraphics[width=0.9\columnwidth]{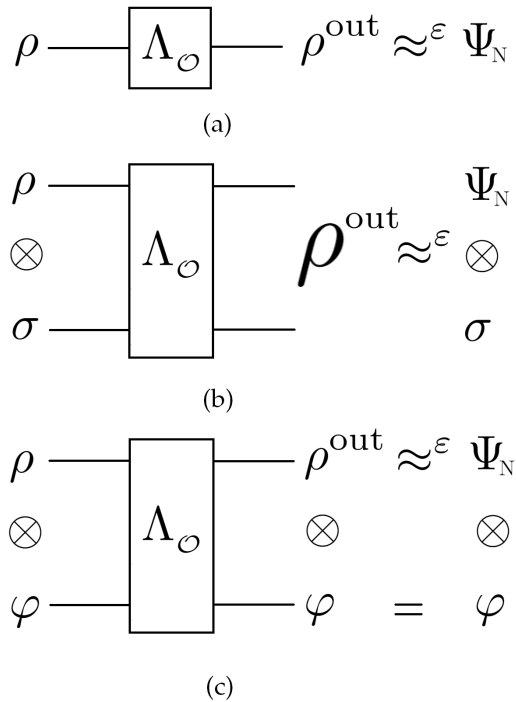}
     \caption{Different schemes for one-shot coherence distillation via incoherent operation $\Lambda_{\mathcal{O}}$. (a) Standard coherence distillation. 
The output state $\rho^\text{out}=\Lambda_{\mathcal{O}}(\rho)$ is close to the $N$-dimensional maximally coherent state $\Psi_{\mathrm{N}}$ with purified distance $\mathrm{P}(\rho^\text{out} , \Psi_{\mathrm{N}}) \le \varepsilon$.
     (b) Catalytic coherence distillation with global smoothing. The joint output state $\rho^{out}=\Lambda_{\mathcal{O}}(\rho\otimes \sigma)$ is close to $\Psi_{\mathrm{N}}\otimes \sigma $ with purified distance
  $\mathrm{P}(\rho^\text{out},\Psi_{\mathrm{N}}\otimes \sigma )  \le \varepsilon$. (c) Catalytic coherence distillation with perfect catalysts. The catalyst $\varphi$ is required to be pure and precisely unchanged, and the output state $\rho^\text{out}$ in the distilled system satisfies
 $\mathrm{P}\left(\rho^\text{out},\Psi_N\right)\le\varepsilon$.} 
\end{figure}
To make the picture clear, we show different schemes of one-shot coherence distillation in FIG.~1.(a)-(c).

The paper is organized as follows: Sec.~\ref{pre} gives some preliminaries and notations, where we review the basics of coherence resource theory and introduce the standard distillation problem. Sec.~\ref{sec:level2} defines the distillable coherence with unrestricted catalysts and global smoothing. We show the existence of ``coherence embezzling phenomenon'' with two different embezzling protocols. In Sec.~\ref{sec:level3}, we give an amended definition, the distillable coherence with catalysts of a restricted dimension, and calculate a lower bound of this quantity. In Sec.~\ref{perfect}, we propose another amended definition, the distillable coherence with perfect catalysts, and find that catalysts become basically helpless in pure state distillation via IO and SIO under certain smoothing restriction. In Sec.~\ref{sec:level4}, we summarize our results and discuss their connection with some recent work, as well as some open questions.


\section{Preliminary}\label{pre}
In this section, we review the resource theory of quantum coherence and standard one-shot coherence distillation problems. 
Throughout this paper we denote the predetermined computational basis as $I = \{\ket{i}\}_{i=1}^{d}$ in a $d$-dimensional Hilbert space $\mathcal{H}_d$, and all concepts related with coherence are defined with respect to this basis. Incoherent states are defined as $\delta = \sum_{i=1}^{d}p_i\ket{i}\bra{i}$,
where $\{p_i\}$ is a probability distribution satisfying $\sum_{i=1}^{d}p_i=1 $. The set of incoherent states is denoted as $\mathcal{I}$. In the resource theory of quantum coherence, free operations are defined with respect to different operational meanings~\cite{Aaberg2006,baumgratz14,Chitambar16prl, Marvian16,winter16,Yadin16,Chitambar16prl}. Here we only introduce IOs~\cite{Aaberg2006,baumgratz14}, SIOs~\cite{winter16,Yadin16} and PIOs~\cite{Chitambar16prl}, which are used in our distillation process. IOs~\cite{baumgratz14} are completely positive trace preserving (CPTP) maps $\Lambda$ admitting a Kraus operator representation, $\Lambda(\rho) = \sum_n K_n\rho K_n^\dag$, such that the $\{K_n\}$s  are ``incoherent-preserving'' operators satisfying ${K_n\delta K_n^\dag}/{\mathrm{Tr}\left[ K_n\rho K_n^\dag\right]} \in \mathcal{I}$ for all $n$ and all $\delta\in\mathcal{I}$.  
SIOs are CPTP maps $\Lambda(\rho) = \sum_n K_n\rho K_n^\dag$ whose Kraus operator representations satisfy that both $\{K_n\}$ and $\{K_n^\dag\}$ are incoherent-preserving operators \cite{winter16}. PIOs are CPTP maps that allow a ``free dilation'', that is, they can be implemented by adding incoherent ancillary states, applying a jointly incoherent unitary, performing incoherent measurements, and doing classical post-processing~\cite{Chitambar16prl}. The relations of these incoherent operations classes are
\begin{equation}\label{relations}
\text{PIO} \varsubsetneq \text{SIO} \varsubsetneq \text{IO}.
\end{equation}
As a result to be used later, the mixture of incoherent unitary permutations,
\begin{equation}
\Lambda(\rho)=\sum_i P_{\pi_i}\rho P_{\pi_i}^\dagger,
\end{equation}
where $P_{\pi_i}=\sum_{j=1}^n \ket{\pi_i(j)}\bra{j}$ represents a permutation $\pi_i$, is a PIO, which can be directly verified by the definition. Therefore by Eq.~\eqref{relations} it is also an SIO and IO channel. 

The one-shot coherence distillation problem characterizes the maximal resource that can be distilled from an initial state with an allowed error $\varepsilon$. The definition of one-shot coherence distillation is as follows.

\begin{definition}[Regula \emph{et al.}~\cite{Regula-one-shot}, Zhao \emph{et al.}~\cite{Zhao19one}]\label{def:1}
    The one-shot distillable coherence of a state $\rho$ via incoherent operations from the class $\mathcal{O}$ is defined as
    \begin{equation}
      C^{\mathrm{\varepsilon}}_{\mathcal{O},d}(\rho) = \max_{\Lambda \in \mathcal{O}} \{{\log \mathrm{N}} :  {\mathrm{P}(\Lambda(\rho)} , {\Psi_{\mathrm{N}}}) \le \varepsilon\},
    \end{equation}
    where $\ket{\Psi_{\mathrm{N}}} = \frac{1}{\sqrt N}\sum_{i=1}^{N} \ket{i}  $ is the maximally coherent state of dimension $\mathrm{N}$. $\mathrm{P}(\rho,\sigma)$ is the purified distance $\mathrm{P}(\rho,\sigma)=\sqrt{1-\mathrm{F}(\rho,\sigma)^2}$ with fidelity $\mathrm{F}(\rho,\sigma) = \Tr\sqrt{\sqrt{\rho}\sigma\sqrt{\rho}}$.
\end{definition}
Throughout this work, logarithms are base 2. When expressing the purified distance and fidelity, if the two states being measured are pure states, say $\ket{\psi},\ket{\phi}$, we use the notation $\mathrm{P}({\psi},{\phi})$ and the alike, which can be understood as $\psi=\ket{\psi}\bra{\psi}$.

As for IO, in order to quantify $C^{\mathrm{\varepsilon}}_{\mathrm{IO},d}$, we introduce the \textit{smooth min-entropy of coherence} $C_{\min}^\varepsilon(\rho)$
\begin{equation}
  \label{Cmin}
  C_{\min}^\varepsilon(\rho)
    = \max_{\rho'\in B^\varepsilon(\rho)} \min_{\delta\in\mathcal{I}} D_{\min}(\rho\|\delta),
\end{equation}
where $B_\varepsilon(\rho)=\{\rho':\mathrm{P}(\rho',\rho)\le\varepsilon\}$ and $D_{\min}(\rho\|\sigma)=-\log \mathrm{F}(\rho,\sigma)^2=\tilde{D}_{1/2}(\rho\|\sigma)$~\cite{2013MartinRenyi}. Here $\tilde{D}_{1/2}$ is the special case of the sandwiched quantum $\alpha$-R\'enyi divergence $\tilde{D}_{\alpha}$ with $\alpha=1/2$ \cite{2013MartinRenyi,wilde2014strong}
\begin{equation}\label{alphadivergence}
    \tilde{D}_{\alpha}(\rho\|\sigma)=\frac{1}{\alpha-1}\log\left(\Tr\left[\left(\sigma^{\frac{1-\alpha}{2\alpha}}\rho\sigma^{\frac{1-\alpha}{2\alpha}}\right)^{\alpha}\right]\right).
  \end{equation}
The properties of the smooth entropy of coherence and its relationship with other coherence monotones are also explored in \cite{Liu18}. 
In Zhao \emph{et al.}~\cite{Zhao19one}, it is shown that 
\begin{equation}
 C^{\mathrm{\varepsilon}}_{\mathrm{IO,d}}(\rho)\approx C_{\min}^{\varepsilon'}(\rho),
\end{equation}
with $\varepsilon'\in[\frac{\varepsilon}{2},\sqrt{\varepsilon(2-\varepsilon)}]$.

As for SIO, Zhao \emph{et al.}~\cite{Zhao19one} show that there exists SIO bound coherence that cannot be distilled in both asymptotic and non-asymptotic cases. Afterward, Lami \emph{et al.}~\cite{PhysRevLett.122.150402} show that SIO bound coherence is actually generic, and also give a explicit formula to completely characterize the asymptotic distillable coherence under SIO as well as PIO. The problem of one-shot distillable coherence under SIO and PIO is still left open.

When the state to be distilled is restricted as a pure state, a complete characterization of the one-shot distillable coherence is given by Regula \emph{et al.}~\cite{Regula-one-shot}. They show that the operation classes  $\text{MIO, DIO, IO, SIO}$ all have the same power in the task of pure state coherence distillation and can be characterized by a quantum hypothesis testing problem. In the case of zero-error distillation, the result is
\begin{equation}
    C^0_{\mathcal{O},d}(\phi)=\log\floor{1 / \phi_{\max}},
\end{equation}
where $\ket\phi=\sum_{i=1}^n\sqrt{\phi_i}\ket i$ and $\phi_{\max} =\max\{\phi_i\}$ .


\section{Distillable Coherence with Unrestricted Catalysts}\label{sec:level2}



Different from the standard coherence distillation or assisted coherence distillation, we introduce and study catalytic coherence distillation. The procedure of catalytic coherence distillation is that, given a state $\rho$, we can choose another state $\sigma_{\mathrm{M}}$ as catalysts, and apply a jointly free operation on the total system. We require that the catalysts remain roughly unchanged in this process. Thus the final state should be close to a product state of a maximally coherent state and the unchanged catalysts $\Psi_{\mathrm{N}}\tens \sigma_\mathrm{M}$, within a purified distance $\varepsilon$. Similar to Def.~\ref{def:1}, it is straightforward to give a definition of catalytic distillable coherence as below.
\begin{definition}[Distillable Coherence with Unrestricted Catalysts]\label{def:unrestrict}
    The catalytic distillable coherence of a state $\rho$ via the free operation class $\mathcal{O}$ with unrestricted catalysts is defined as
    \begin{equation}
    \begin{split}
      &C^{\mathrm{\varepsilon}}_{\mathrm{\mathcal{O},c}}(\rho) =\\&  \max_{\Lambda \in \mathcal{O}}\max_{dim(\sigma_\mathrm{M}) < \infty}\{{\log \mathrm{N}} : {\mathrm{P}(\Lambda(\rho \tens \sigma_\mathrm{M})} , {\Psi_{\mathrm{N}}}\tens \sigma_\mathrm{M})\leq \varepsilon\}.
    \end{split}
    \end{equation}
\end{definition}
In this definition, we only require that the dimension of the catalytic state is finite. In the following, we show that this definition leads to a strange phenomenon that one can catalytically distill as much coherence as he wants, from any initial state and with respect to an arbitrarily small $\varepsilon$. This kind of phenomenon is first found in the resource theory of entanglement~\cite{Embez} and denoted as the embezzling phenomenon, and later generalized to other resource theories such as thermodynamics~\cite{EmbezThermo}. In our scenario, we formulate this phenomenon as the following theorem.

\begin{theorem}\label{th:unrestrict}
	For any quantum state $\rho$, an arbitrarily large integer $N$ and any $\varepsilon> 0$, there exists a PIO (therefore also SIO, IO) channel $\Lambda$ and a catalytic state $\sigma_M$ of dimension $M<\infty$ such that
	\begin{equation}
		\mathrm{P}\left(\Lambda(\rho \tens \sigma_M),\Psi_{\mathrm{N}}\tens \sigma_M\right)\leq\varepsilon.
	\end{equation}
	\end{theorem}
As a result of Thm.~\ref{th:unrestrict}, the quantity in Def.~\ref{def:unrestrict} diverges to infinity, which is called the coherence embezzling phenomenon here. To prove this theorem, we give two different protocols. The first one applies the convex-split lemma proposed in \cite{ConvexSplit}. The second one uses an embezzling state modified from the entanglement embezzling state \cite{Embez}.

\subsection{A Protocol Using the Convex-Split Lemma}\label{sec:level1}

The convex-split lemma is first proposed in \cite{ConvexSplit} as a mathematical tool inspired by classical communication theory. This lemma has also been applied to the study of catalytic decoupling \cite{CataDec} and quantifying resources with resource destroying maps in the assistance with catalysts \cite{CataRes}.

Before presenting the convex-split lemma, we give the definition of max R\'enyi divergence. For two quantum states $\omega_Q$ and $\sigma_Q $ satisfying $supp(\omega_{Q} ) \subset supp( \sigma_Q)$,  the max R\'enyi divergence $D_{max}(\omega_{Q} ||\sigma_Q)$ is defined as
\begin{equation}
 D_{max}(\omega_{Q} ||\sigma_Q)=\min\left\{\log\lambda:\lambda\sigma_Q\ge\omega_{Q}\right\}.
\end{equation}

\begin{lemma}[Convex-split lemma \cite{ConvexSplit}]\label{le:cs}
 Let n be an integer, and $k = D_{max}(\omega_{Q} ||\sigma_Q)$. 
   Consider the following quantum state
    \begin{equation}
    \begin{split}
        &\tau_{Q_1Q_2Q_3...Q_n} := \\ &\frac{1}{n} \sum_{j=1}^{n} \sigma_{Q_1} \tens \sigma_{Q_2} \tens ... \tens \sigma_{Q_{j-1}}\tens \omega_{Q_j} \tens \sigma_{Q_{j+1}}\tens ... \tens \sigma_{Q_n},
    \end{split}
    \end{equation}
on $n$ ordered registers $Q_1,Q_2,...Q_n$, where $\forall j \in [n], \omega_{Q_j} = \omega_{Q}$ and $ \sigma_{Q_j} = \sigma_Q$. Then,
    \begin{equation}
        \mathrm{P}(\tau_{Q_1Q_2Q_3...Q_n}, \sigma_{Q_1} \tens ... \tens \sigma_{Q_n}) \leq \sqrt{\frac{2^k}{n}}.
    \end{equation}
    In particular, let $n= \lceil\frac{2^k}{\gamma^2}\rceil$ for some $\gamma > 0$, then,
    \begin{equation}
        \mathrm{P}(\tau_{Q_1Q_2Q_3...Q_n}, \sigma_{Q_1} \tens ... \tens \sigma_{Q_n}) \leq \gamma.
    \end{equation}
\end{lemma}
This lemma mainly gives a construction of a quantum state $\tau$ which is close to $\sigma_Q^{\tens n}$ by the following steps: (1) Choose one out of $n$ registers to put in the quantum state $\omega_Q$, and fill the other $(n-1)$ registers each with a state $\sigma_Q$. Since we have $n$ registers to put $\omega_Q$ in, there are $n$ different result states. 
(2) Mix these $n$ different states up with an equal probability, we then get the desired $\tau$.

Intuitively, the difference between $\omega$ and $\sigma$ becomes blurred when mixing all the states up. The information of $\omega$ is approximately erased and the state in every register looks like $\sigma$. Now we can prove Thm.~\ref{th:unrestrict} with a protocol based on Lemma~\ref{le:cs}.\\

\emph{Proof of Theorem 1}.
To begin with, for $N>dim(\rho)$, we can always incoherently embed $\rho$ into a Hilbert space $\mathcal{H}_N$ with dimension $N$ by adding trivial degrees of freedom, with $N$ the dimension of our desired maximally coherent state. 
For simplicity, we 
still denote the state 
as $\rho$ after embedding it into a larger Hilbert space hereafter. 
We first apply a twirling channel 
\begin{equation}
    \mathcal{T}(\rho) = \frac{1}{N!}\sum_{i=1}^{N!}P_i \rho P_i,
\end{equation} 
where $\left\{P_i\right\}$ form the incoherent permutation operator group on $\mathcal{H}_N$. After twirling, every label of the output state is symmetric, so the output state can without loss of generality be written as \cite{zhou2017polynomial},
\begin{equation}\label{TwirlState}
	\mathcal{T}(\rho) = \\(1-t^2)\ket{\Psi_N}\bra{\Psi_N} + t^2\frac{\mathbb{I}-\ket{\Psi_N}\bra{\Psi_N}}{N-1},
\end{equation}
where $\Psi_N$ is the maximally coherent state in Hilbert space $\mathcal{H}_N$ and $t\in \left[0,1\right]$ is a parameter of $\rho$ which equals to $\mathrm{P}(\mathcal{T}(\rho),\Psi_N)$. Intuitively, a smaller $t$ corresponds to a more coherent $\rho$, and also easier for us to distill $\ket{\Psi_N}$ from $\rho$. In the following, we denote  $\mathcal{T}(\rho)$ as $\rho^\mathcal{T}$ for simplicity. Without loss of generality we can assume $t>\varepsilon$, otherwise the problem of distillation becomes trivial since $\rho^\mathcal{T}$ is already $\varepsilon$-close to \ket{\Psi_N}.
	
Now we take
\begin{equation}
        \sigma_N:=(1-\delta^2)\ket{\Psi_N}\bra{\Psi_N}+\delta^2\frac{\mathbb{I}-\ket{\Psi_N}\bra{\Psi_N}}{N-1},
\end{equation}
for some $0 < \delta < \varepsilon$ to be chosen later. For a given finite $\varepsilon$ such $\delta$ always exists.
Let $k:= D_{\max}(\rho^\mathcal{T} || \sigma_N)$. Obviously $supp(\rho^\mathcal{T}) \subset supp(\sigma_N)$, since the latter spans the whole $N$-dimensional space, guaranteeing $k$ to be finite. Let $n:= \lceil \frac{2^k}{\gamma^2} \rceil$ with $\gamma$ to be chosen later, and construct the convex-split channel $\Lambda$ as follows ($N_i$ represents the $i$'th quantum register as in Lemma~\ref{le:cs}, with its Hilbert space of dimension $N$) 
\begin{equation}
\begin{aligned}
    &\Lambda(\rho^\mathcal{T} \tens \sigma_N^{\tens n-1}) =\Lambda(\rho^\mathcal{T}_{N_1} \tens \sigma_{N_2} \tens \sigma_{N_3} \tens ...\tens \sigma_{N_n}) \\ &= \frac{1}{n} \sum_{j=1}^{n} \sigma_{N_1}  \tens ... \tens \sigma_{N_{j-1}} \tens\rho^\mathcal{T}_{N_j} \tens \sigma_{N_{j+1}}\tens ... \tens \sigma_{N_n}.
\end{aligned}
\end{equation}
This is just the state $\tau$ in Lemma \ref{le:cs}, so we have
    \begin{equation}
        \mathrm{P}(\Lambda(\rho^\mathcal{T} \tens \sigma_N^{\tens n-1}),\sigma_N^{\tens n} ) \leq \gamma.
    \end{equation}
    On the other hand, by the property of purified distance,
    \begin{equation}
        \mathrm{P}(\sigma_N^{\tens n} , {\Psi_N} \tens \sigma_N^{\tens n-1}) = \mathrm{P}(\sigma_N , {\Psi_N}) = \delta.
    \end{equation}
    Invoking the triangular inequality of purified distance, we have
    \begin{equation}
    \mathrm{P}(\Lambda(\rho^\mathcal{T} \tens \sigma_N^{\tens n-1}),{\Psi_N} \tens \sigma_N^{\tens n-1}) \leq \gamma + \delta.
    \end{equation}
Therefore, we can choose $\gamma$, $\delta$ to satisfy $\gamma + \delta \leq \varepsilon$ at the expense of a large enough $n$, so as to achieve the required precesion of distillation. Note that, the whole catalytic system is $\sigma_M = \sigma_N^{\tens n-1}$ with dimension $M = N^{n-1}$, which is finite.
    
We summarize the protocol as follows:
\begin{align}
&\rho\tens\sigma_N^{\tens n-1} \xrightarrow{\text{Twirling on $\rho$}}
\rho^\mathcal{T}\tens\sigma_N^{\tens n-1}\\ &\xrightarrow{\text{Convex-split channel}}^\varepsilon
\ket{\Psi_N}\bra{\Psi_N}\tens\sigma_N^{\tens n-1}.
\end{align}
Both the twirling channel and the convex-split channel are mixtures of incoherent permutations, so the whole protocol is PIO. 
    
Now let's calculate $n= \lceil\frac{2^k}{\gamma^2}\rceil$, which determines the number of $\sigma_N$ needed as catalysts in this protocol. We have
\begin{equation}
\begin{aligned}
    2^k = 2^{D_{\max}(\rho^\mathcal{T}||\sigma_N)} = \min{\{\lambda : \lambda\sigma_N \ge \rho^\mathcal{T}}\},
\end{aligned}
\end{equation}
that is, to minimize $\lambda$ in
\begin{equation}    
    \begin{aligned}
    &\lambda  \left((1-\delta^2)\ket{\Psi_N}\bra{\Psi_N}+\delta^2\frac{\mathbb{I}-\ket{\Psi_N}\bra{\Psi_N}}{N-1}\right) \\ 
    &\ge (1-t^2)\ket{\Psi_N}\bra{\Psi_N} + t^2\frac{\mathbb{I}-\ket{\Psi_N}\bra{\Psi_N}}{N-1}.
    \end{aligned}
\end{equation}
The result is (note that $t>\delta$), $\lambda\geq\frac{t^2}{\delta^2}$.
This means that $(n-1)$ copies of $\sigma_N$ as catalysts with $n=\ceil{\cfrac{t^2}{\gamma^2\delta^2}}$ are enough to catalytically distill ${\Psi_N}$. By choosing $\gamma = \delta = \frac\varepsilon2$, this corresponds to a catalyst of dimension $M$ which satisfies
\begin{equation}\label{BoundCS}
    M = N^{\ceil{16t^2\varepsilon^{-4}}-1}.
\end{equation}

For the case where $N<dim(\rho)$, we can always first catalytically distill a maximally coherent state with dimension $N'\geq dim(\rho)$ with the procedures above, and then discard some of its dimensions to obtain a maximally coherent state with a dimension $N$. This completes the proof of Thm.~\ref{th:unrestrict}.


\subsection{A Protocol with the Embezzling State}\label{pro:embezzling}

The embezzling state is first found in entanglement transformation~\cite{Embez}, and later extended to thermodynamics~\cite{EmbezThermo}. We modified  the protocol in \cite{Embez} and apply it to  the framework of coherence. The protocol can incoherently transform any state to an arbitrarily large maximally coherent state within an arbitrary precision, given a catalytic state with a large enough yet finite dimension.\\

\emph{Alternative Proof of Theorem 1}. In this protocol, the family of catalysts are chosen as:
\begin{equation}
\ket{\sigma_M} = \frac{1}{\sqrt{C(M)}}\sum_{j=1}^M\frac{1}{\sqrt{j}}\ket{j},
\end{equation}
where $C(M)=\sum_{j=1}^{M}\frac{1}{j}$ is a normalization factor. This family of catalysts is called the ``embezzling states''. Our target state is  
\begin{equation}
    \ket{\sigma_M}\otimes\ket{\Psi_N}=\cfrac1{\sqrt{C(M)}}\sum_{j=1}^M\sum_{i=1}^N\cfrac1{\sqrt{Nj}}\ket{j}\ket{i}.\label{eq:27}
\end{equation}
Suppose we are able to create a state $\ket{\omega}$ which has the same coefficients in the incoherent basis as $\ket{\sigma_M}\tens\ket{\Psi_N}$, but rearranged in a non-increasing order. Then, with an incoherent permutation operator, we can convert $\ket{\omega}$ into our desired $\ket{\sigma_M}\tens\ket{\Psi_N}$. Interestingly, thanks to the special structure of $\ket{\sigma_M}$, one can show that, without any extra operation, the fidelity between $\ket{\sigma(M)}\tens\ket{1}$ and $\ket{\omega}$ is already arbitrarily close to 1, as $M$ approaches infinity. This means $\ket{\omega}$ is at our hand from the very beginning. Now we explicitly show this. Write the state $\ket{\omega}$ as follows,
\begin{equation}
    \ket\omega=\sum_{j=1}^M\sum_{i=1}^N \omega_{ji} \ket{j}\ket{i},
\end{equation}
 where the coefficients of $\ket{\omega}$ satisfy $\omega_{11}\ge\omega_{21}\ge...\ge\omega_{M1}\ge\omega_{12}\ge...\ge\omega_{MN}$, a reordering of those of $\ket{\sigma_M}\otimes\ket{\Psi_N}$. We always refer to $\ket\omega$'s coefficients in this order. Compared with Eq.~\eqref{eq:27}, we see the first $N$ coefficients of $\ket{\omega}$ equal to $1/{\sqrt{NC(M)}}$, the following $N$ equal to $1/\sqrt{2NC(M)}$, etc. Then the first $M$ coefficients of $\ket\omega$ can be written as
\begin{equation}
    \omega_{j1}=1/{\sqrt{\ceil{j/N}NC(M)}}\le1/{\sqrt{jC(M)}}.\label{eq:omega_ij}
\end{equation}
On the other hand,
\begin{align}
    \ket{\sigma_M}\otimes\ket{1}&=\cfrac1{\sqrt{C(M)}}\sum_{j=1}^M\cfrac1{\sqrt{j}}\ket{j}\otimes \ket{1}.
\end{align}
Thus, we have
\begin{equation}
\begin{aligned}
&\mathrm{F}(\ket{\sigma_M}\otimes\ket{1},\ket{\omega}) =\sum_{j=1}^M\cfrac{\omega_{j1}}{\sqrt{jC(M)}}\ge\sum_{j=1}^M\omega_{j1}^2\\
&\ge\sum_{i=1}^{\lfloor M/N \rfloor}\sum_{j=1}^N \cfrac1{iNC(M)}=\cfrac{\sum_{i=1}^{\lfloor M/N \rfloor}\cfrac1i}{\sum_{i=1}^M\cfrac1i}\ge 1-\cfrac{1+\log N}{1+\log M},
\end{aligned}
\end{equation}
where the first inequality is by equation \eqref{eq:omega_ij}, the second is by directly calculating the first $N\cdot\floor{M/N}$ terms of the sum, and the last uses definite integral to bound the sum of sequences.

By applying an appropriate incoherent permutation unitary $U_P$ which satisfies $\ket{\sigma_M}\otimes\ket{\Psi_N} = U_P\ket{\omega}$, and using the unitary invariance of fidelity, we have
\begin{equation}
\begin{aligned}
    &\mathrm{F}(U_P(\ket{\sigma_M}\otimes\ket{1}),\ket{\sigma_M}\otimes\ket{\Psi_N})\\
    =& \mathrm{F}(U_P(\ket{\sigma_M}\otimes\ket{1}),U_P\ket{\omega})\\
    =& \mathrm{F}(\ket{\sigma_M}\otimes\ket{1},\ket{\omega})\\
    \ge& 1-\cfrac{1+\log N}{1+\log M},
\end{aligned}
\end{equation}
therefore,
\begin{align}
    \mathrm{P}\left(U_P(\ket{\sigma_M}\otimes\ket{1}),\ket{\sigma_M}\otimes\ket{\Psi_N}\right)\le \sqrt{2\cfrac{1+\log N}{1+\log M}}.
\end{align}
As a result, given $N$, $\varepsilon$, in order to embezzle $\ket{\Psi_N}$ with a purified distance smaller than $\varepsilon$, we only need to choose the dimension $M$ of the catalyst $\ket{\sigma_M}$ to satisfy
\begin{equation}\label{BoundEB}
    M=\ceil{(2N)^{2\varepsilon^{-2}}/2}.
\end{equation}
We summarize the whole protocol as follows:
\begin{align}
   & \ket{\sigma_M}\otimes\ket{\rho} \xrightarrow{\text{discard\ and\ re-prepare}} \ket{\sigma_M}\otimes\ket{1} \approx^\varepsilon \ket{\omega}\\ &\xrightarrow{\text{incoherent permutation } U_P} \ket{\sigma_M}\otimes\ket{\Psi_N}.
\end{align}
The first step includes a partial trace and adding free ancillary states, while the second is an incoherent permutation. Hence, the whole protocol is a PIO. This completes the proof.


Note that, in this protocol, the initial state $\rho$ becomes completely irrelevant. It seems quite weird to call this a distillation protocol.

\begin{remark}
One possible explanation for the coherence embezzling phenomenon is that by allowing global smoothing, the distilled coherence is not solely from the system of interest, but also ``embezzled'' from the catalysts via the allowable fidelity error. Since the catalysts can be arbitrarily large, embezzlement of a relatively small amount of coherence (which is large compared with the state of interest) from them without causing big fidelity error can be possible.
\end{remark}


\section{Distillable Coherence with Catalysts of a Restricted Dimension}\label{sec:level3}
As shown in Sec.~\ref{sec:level2}, the catalytically distillable coherence diverges under Def.~\ref{def:unrestrict}. There may be several ways to ease this embezzling phenomenon. Here we consider one possible modification, in which the coherence distillation procedure can only use catalysts of a restricted finite dimension, as follows,
\begin{definition}[Distillable Coherence with Restricted Catalysts]\label{def:restrict}
    The catalytic distillable coherence of a state $\rho$ via free operation class $\mathcal{O}$ with catalysts of dimension no more than $M$ is defined as
    \begin{equation}
    \begin{split}
      &C^{\mathrm{\varepsilon,\mathrm{M}}}_{\mathrm{\mathcal{O},rc}}(\rho) = \\& \max_{\Lambda \in \mathcal{O}}\max_{dim(\sigma_\mathrm{M}) \le \mathrm{M}}\{{\log \mathrm{N}} : {\mathrm{P}(\Lambda(\rho \tens \sigma_\mathrm{M})} , {\Psi_{\mathrm{N}}}\tens \sigma_\mathrm{M})\leq \varepsilon\}.
    \end{split}
    \end{equation}
\end{definition}
We also define the catalytically distillable coherence given restricted catalysts with respect to a specific protocol $\mathcal{P}$, denote as $C^{\mathrm{\varepsilon,\mathrm{M}},\mathcal{P}}_{\mathrm{\mathcal{O},rc}}(\rho)$, which represents the efficiency of the protocol as well as serves as a lower bound of $C^{\mathrm{\varepsilon,\mathrm{M}}}_{\mathrm{\mathcal{O},rc}}(\rho)$. 
With two specific protocols proposed in Sec.~\ref{sec:level2} to show embezzling phenomenon, we can determine a lower bound of $C^{\mathrm{\varepsilon,\mathrm{M}}}_{\mathrm{\mathcal{O},rc}}(\rho)$. We denote the protocol with the convex-split lemma as $\mathcal{P}_{cs}$ and the one with embezzling state as $\mathcal{P}_{eb}$.

Here we provide a lower bound given by $\mathcal{P}_{eb}$.
\begin{theorem}\label{th:restrict} Given an integer $M$ and a positive number $\varepsilon$ such that $\varepsilon^2\log(M-1) / 4\ge1$, the distillable coherence with catalysts of dimension no more than $M$ is lower bounded by
\begin{equation}
    C^{\mathrm{\varepsilon,\mathrm{M}}}_{\mathrm{,rc}}(\rho) \ge  \frac12\varepsilon^2 \left(\log(M-1)+1\right)-2,
\end{equation}
for $\mathcal{O} = \text{\{IO, SIO, PIO\}}$.
\end{theorem}

\begin{proof}
    Given $\varepsilon$ and $M$, from Eq.~\eqref{BoundEB}, we know that one can distill ${\Psi_N}$ from any initial state with catalysts of dimension $M$ as long as $N$ satisfies
    \begin{equation}
        M\ge\ceil{(2N)^{2\varepsilon^{-2}}/2}.
    \end{equation}
    For this inequality to holds, $N$ only needs to satisfy
    \begin{equation}
        M\ge{(2N)^{2\varepsilon^{-2}}/2}+1,
    \end{equation}
    which is equivalent to
    \begin{equation}
        N\le2^{\left(\frac12\varepsilon^2 \left(\log(M-1)+1\right)-1\right)}.
    \end{equation}
    Taking into consideration that $N$ must be integer, we get an achievable $N^*$ as follows,
    \begin{equation}
        N^* = \floor{2^{\left(\frac12\varepsilon^2 \left(\log(M-1)+1\right)-1\right)}}.
    \end{equation}
    Therefore,
    \begin{align}
        C^{\mathrm{\varepsilon,\mathrm{M}}}_{\mathrm{\mathcal{O},rc}}(\rho) \ge \log N^* \ge \frac12\varepsilon^2 \left(\log(M-1)+1\right)-2,
    \end{align}
    where the last inequality use the condition $\varepsilon^2\log(M-1) / 4\ge1$.
\end{proof}

Then we give a rough analysis of the efficiency of both $\mathcal{P}_{cs}$ and $\mathcal{P}_{eb}$ when $\varepsilon^2 \ll 1, \varepsilon^2\log M\gg 1$. In this case, we can ignore the ceiling function in Eq.~\eqref{BoundCS} and Eq.~\eqref{BoundEB} and derive the following estimations
\begin{align}
       &C^{\mathrm{\varepsilon,\mathrm{M}},\mathcal{P}_{cs}}_{\mathrm{\mathcal{O},rc}}(\rho) \gtrsim
    \cfrac{\log M}{16t^2\varepsilon^{-4}-1} \approx \frac{1}{16}\frac{\varepsilon^4}{t^2}\log M, \label{Pcs}\\
    
    &C^{\mathrm{\varepsilon,\mathrm{M}},\mathcal{P}_{eb}}_{\mathrm{\mathcal{O},rc}}(\rho) \gtrsim
    \frac12\varepsilon^2(\log M+1) -1 \approx \frac12\varepsilon^2\log M, \label{eq:P_eb}
\end{align}
for $\mathcal{O} = \text{\{IO, SIO, PIO\}}$. $t$ in Eq.~\eqref{Pcs} is a parameter of $\rho$ as described in Sec.~\ref{sec:level1}. Since $t$ also depends on the dimension of the maximally coherent state to be distilled, we cannot directly calculate it in Eq.~\eqref{Pcs}. Nevertheless, we have assumed that $t>\varepsilon$, so the bound in Eq.~\eqref{Pcs} cannot be better than $\frac{1}{16}\varepsilon^2\log M$, which means in general the bound of $\mathcal{P}_{cs}$ cannot be better than that of $\mathcal{P}_{eb}$.

\begin{remark}
In Eq.~\eqref{eq:P_eb}, the amount of distillable coherence becomes completely independent of the initial state. It seems quite weird to call this a distillation protocol. Nevertheless, it indeed satisfies Def.~\ref{def:restrict}. This fact also implies that the lower bound in Thm.~\ref{th:restrict} is quite loose, since it does not utilize the initial coherence of $\rho$. A better characterization of $C^{\mathrm{\varepsilon,\mathrm{M}}}_{\mathrm{\mathcal{O},rc}}(\rho)$ is left for further research.
\end{remark}

\section{Distillable Coherence with Perfect Catalyst}\label{perfect}
In this section, we propose the requirement ``perfect catalysts'', as another approach to deal with the coherence embezzling phenomenon. Here, the catalytic system must be precisely unchanged after the whole procedure. In particular, we consider a pure state as perfect catalysts, and the output state in the catalytic system is required to have precisely unit fidelity with the input catalysts, and a small error is only allowable in the system of distillation.

Although the restriction is strict,  the consideration behind this definition is that, as for a real catalysts, it should be able to work again and again with the same effectiveness. However, in the unrestricted catalyst scenario, every time the procedure is conducted, a small error may accumulate in the catalytic system. Therefore, it becomes questionable whether the system can still serve as catalysts after several rounds of such procedure. On the contrary, here in the perfect catalyst scenario, no error ever accumulates in the catalytic system, so it could catalyze the procedure arbitrarily many times without a decline in effectiveness. Further more, by requiring catalysts to be pure, we ensure that even their purified system do not change, which guarantees the catalysts are physically unchanged. We formalize the definition of perfect catalytic coherence distillation as follows.



\begin{definition}[Perfect catalyst]\label{def:perfect}
Let $\mathcal{H}_A$, $\mathcal{H}_C$ denote the Hilbert spaces of the system of interest and the catalytic system, respectively. For any quantum state $\rho$ on $\mathcal{H}_A$, the one-shot distillable coherence with perfect catalysts via free operation class $\mathcal{O}$ is defined as
\begin{equation}
\begin{aligned}
    &{\mathrm{C}}^\varepsilon_{\mathcal{O},\text{pc}}(\rho)=\max_{N}\ \log{N}, \\
    \text{s.t.}\quad&\exists\
     \Lambda\in\mathcal{O},\
     \ket\omega\in\mathcal{H}_C,\
    \rho^\text{out}\in\mathcal{D}(\mathcal{H}_A),\\
    &\mathrm{Tr}_A\left(\Lambda\left(\rho\tens\ket\omega\bra\omega\right)\right)=\ket\omega\bra\omega,\\
    &\mathrm{Tr}_C\left(\Lambda\left(\rho\tens\ket\omega\bra\omega\right)\right)=\rho^\text{out},\\
    &\mathrm{P}\left(\rho^\text{out},\Psi_N\right)\le\varepsilon.
\end{aligned}    
\end{equation}
\end{definition}
This procedure can be informally written as  
\begin{equation}
\rho\tens\ket\omega\xrightarrow{\mathcal{O}}\Psi_N^\varepsilon\tens\ket\omega, 
\end{equation}
where $\ket\omega$ serves as a perfect catalyst, which is shown in FIG.1.(c). 

In fact, the quantity in Def.~\ref{def:perfect} could be quite difficult to calculate. In general, in the resource theory of entanglement or coherence, the conversion problem with perfect catalysts are mostly only solved for pure states, such as in \cite{Jonathan99, Bu16}. Here, instead of fully characterizing $\mathrm{C}^\varepsilon_{\mathcal{O},\text{pc}}(\rho)$, we solve this problem for pure states under the following ``pure state smoothing'' restriction.

\begin{definition}[Perfect catalysts, under pure state smoothing restriction]
For pure state $\ket\phi \in \mathcal{H}_A$, the one-shot distillable coherence with perfect catalysts via operation class $\mathcal{O}$ under pure state smoothing restriction is defined as
\begin{equation}
\begin{aligned}
    &\tilde{\mathrm{C}}^\varepsilon_{\mathcal{O},\text{pc}}(\phi)=\max_{N}\ \log{N}, \\
    \text{s.t.}\quad&\exists\ \Lambda\in\mathcal{O},\ \ket\omega\in\mathcal{H}_C,\
    \ket{\phi^\text{out}}\in\mathcal{H}_A,\\
    &\mathrm{Tr}_A\left(\Lambda\left(\ket\phi\bra\phi\tens\ket\omega\bra\omega\right)\right)=\ket\omega\bra\omega,\\
    &\mathrm{Tr}_C\left(\Lambda\left(\ket\phi\bra\phi\tens\ket\omega\bra\omega\right)\right)=\ket{\phi^\text{out}}\bra{\phi^\text{out}},\\
    &\mathrm{P}\left(\phi^\text{out},\Psi_N\right)\le\varepsilon.
\end{aligned}    
\end{equation}
\end{definition}
That is 
\begin{equation}
    \ket\phi\tens\ket\omega\xrightarrow{\mathcal{O}}\ket{\Psi_\varepsilon}\tens\ket\omega.
\end{equation}
We impose a restriction that the smoothed output state, i.e., $\ket{\phi^\text{out}}$ is pure. When considering exact distillation ($\varepsilon = 0$), this quantity becomes the same with the one in Def.~\ref{def:perfect}.
Similarly we can define a ``pure state smoothing'' version of the standard one-shot distillable coherence (Def. \ref{def:1}) as follows,
\begin{equation}
\begin{aligned}
    &\tilde{\mathrm{C}}^\varepsilon_{\mathcal{O},\text{d}}(\phi)=\max_{N}\ \log{N}, \\
    \text{s.t.}\quad&\exists\ \Lambda\in\mathcal{O},\ \ket{\phi^\text{out}}\in\mathcal{H}_A,\\
    &\Lambda\left(\phi\right)={\phi^\text{out}},\\
    &\mathrm{P}\left(\phi^\text{out},\Psi_N\right)\le\varepsilon.
\end{aligned}    
\end{equation}


Our main result in this section is the following theorem.

\begin{theorem}\label{th:perfect}
For a pure state $\ket\phi$, let $r>1$ denote the number of its non-zero coefficients in the incoherent basis, and let $\varepsilon$ be a positive number that satisfies $ r (r-1)\varepsilon\le 1$, then 
\begin{align}
    \tilde{\mathrm{C}}^\varepsilon_{IO,d}(\phi)
    \le\tilde{\mathrm{C}}^\varepsilon_{IO,pc}(\phi) 
    \le
    \tilde{\mathrm{C}}^{\sqrt{r(r-1)\varepsilon}}_{IO,d}(\phi).
\end{align}
In particular, in the case of exact distillation ($\varepsilon = 0$),
\begin{equation}
    {\mathrm{C}}_{IO,d}(\phi)=
    {\mathrm{C}}_{IO,pc}(\phi).
\end{equation}
Moreover, $IO$ distillation rate is the same with $SIO$
\begin{align}
    \tilde{\mathrm{C}}^\varepsilon_{IO,d}(\phi)&=
    \tilde{\mathrm{C}}^\varepsilon_{SIO,d}(\phi),\\
    
    \tilde{\mathrm{C}}^\varepsilon_{IO,pc}(\phi)&=
    \tilde{\mathrm{C}}^\varepsilon_{SIO,pc}(\phi).
\end{align}
\end{theorem}
This theorem implies that, when restricted to pure state catalysts, IO and SIO have exactly the same operational ability in the task of catalytic coherence distillation, just as they do in the non-catalytic distillation task. And, maybe surprisingly, catalysts provide ``almost'' no advantage for pure state distillation with these two operation classes. By ``almost'' we mean that the distillable coherence for a pure state with or without perfect catalysts is the same, with respect to different smoothing parameters that depend on the dimension of the support of the state.

Before we prove Thm.~\ref{th:perfect}, we introduce two lemmas. To characterize the convertibility of pure state via IO and SIO, a mathematical tool named majorization is applied. For two normalized (in this section ``normalized'' means ``sum to 1'') $n$-element  vectors $p$, $q$ with all elements nonnegative, $p$ majorises $q$ ($p\succ q$) iff 
\begin{equation}
  \sum_{i=1}^k p_i^{\downarrow} \ge \sum_{i=1}^k q_i^{\downarrow},\ \forall 1\le k\le n,   
\end{equation}
where $p^{\downarrow}$ is a permutation of $p$'s elements in the non-increasing order.

The majorization is partial order relation. Sometimes $p\nsucc q$, but there might exist some nonnegative normalized vector $w$ such that $p\tens w \succ q\tens w$. If the latter condition holds, we call $q$ is catalytic-majorized (or trumped) by $p$, denote as $p\succ_T q$. The sufficient and necessary condition for catalytic-majorization is given in \cite{Turgut07, Klimesh07} in the study of catalytic entanglement transformation, and later generalized to coherence in \cite{Bu16}. These conditions are summarized in the following lemma.

\begin{lemma} \label{le:1}\cite{Turgut07} 
For two $n$-element normalized vectors p and q with non-negative elements, where all elements of p are positive and $p^\downarrow\neq q^\downarrow$, the relation  $p\prec_T q$ is equivalent to the following condition:
\begin{equation}
    S_\alpha(p) > S_\alpha(q),\quad \forall \alpha  \in (-\infty,+\infty),
\end{equation}
where 
\begin{equation}
S_\alpha(p)=\cfrac{sign(\alpha )}{1-\alpha }\log{\left(\sum_{i=1}^n p_i^\alpha \right)}
\end{equation}
is the R\'enyi entropy, with appropriate limit taken at $\alpha  = 0$ and $\alpha  = 1$. Specifically, when $q$ has zero element, the condition holds trivially for $\alpha \le 0$, and only the case $\alpha >0$ needs to be verified.
\end{lemma}
The next lemma characterizes the pure state convertibility via IO and SIO.

\begin{lemma} \label{le:2} \cite{winter16} For pure state $\ket{\phi},\ket{\psi}$ with dimension $n$, the necessary and sufficient condition for the transformation $\ket{\phi} \to \ket{\psi}$ to be possible via IO or SIO is
\begin{equation}
    \phi^{\Delta} \prec \psi^{\Delta}.
\end{equation}
Here, $\phi^{\Delta} = \left[\phi_1,...,\phi_n \right]$ where $\ket\phi = \sum_i\sqrt{\phi_i}\ket{i}$ is the expansion in the incoherent basis. It is also the diagonal part of $\Delta\left(\ket{\phi}\bra{\phi}\right)$ where $\Delta$ is the dephasing channel.
\end{lemma}

Equipped with the above two lemmas, now we can prove our main result Thm.~\ref{th:perfect} in this section.

\emph{Proof of Theorem 3}.
Under the pure state smoothing restriction, all possible transformations are between pure states (with or without catalysts). Therefore, as a result of Lemma \ref{le:2}, we must have
\begin{align}
    \tilde{\mathrm{C}}^\varepsilon_{IO,d}(\phi)&=
    \tilde{\mathrm{C}}^\varepsilon_{SIO,d}(\phi),\\
    
    \tilde{\mathrm{C}}^\varepsilon_{IO,pc}(\phi)&=
    \tilde{\mathrm{C}}^\varepsilon_{SIO,pc}(\phi).
\end{align}
Hence in the following we only consider IO for simplicity of notation.

First we consider an extreme case---exact distillation ($\varepsilon = 0$). The maximally coherent state $\ket{\Psi_N}$ with dimension $N \le r$ has the following coefficient vector with r elements
\begin{equation}
    \Psi_N^{\Delta}=[\frac1N,...,\frac1N,0,...,0].
\end{equation}
It's easy to see that $\phi^{\Delta} \prec \Psi_N^{\Delta}$ if and only if 
\begin{equation}
    \phi^{\Delta}_{\max}\equiv \max_i\{\phi^\Delta_i\} \le \frac1N.
\end{equation}
So, by Lemma~\ref{le:2}, the exact distillable coherence of $\ket\phi$ is 
\begin{equation}
    C_{IO,d}\left(\phi\right)=\log\floor{1/\phi^{\Delta}_{\max}}.
\end{equation}

Next consider catalytic distillation. By Lemma~\ref{le:2}, For $\ket\phi \xrightarrow{IO,pc} \ket{\Psi_N}$ to be possible, there must exist some $\ket\omega$ such that $\phi^{\Delta}\tens\omega^{\Delta} \prec \Psi_N^{\Delta}\tens\omega^{\Delta}$, which is just $\phi^{\Delta} \prec_T \Psi_N^{\Delta}$. By Lemma~\ref{le:1}, this is equivalent to
\begin{align}
    &S_\alpha(\phi^{\Delta}) > S_\alpha(\Psi_N^{\Delta}),\quad \forall \alpha>0\\
    \Leftrightarrow\quad &S_\alpha(\phi^{\Delta}) > \log N,\quad \forall \alpha>0\\
    \Leftrightarrow\quad &S_\infty(\phi^{\Delta}) \ge \log N\\
    \Leftrightarrow\quad &-\log \phi^{\Delta}_{\max} \ge \log N,
\end{align}
where the third line use the fact that R\'enyi entropy is continuous and non-increasing for $\alpha>0$. In the above derivation, we have assumed (1) $\phi^{\Delta}$ has only non-zero elements, for we can always truncate the Hilbert space to $\ket\phi$'s support without loss of generality, and (2) $\Psi_N^{\Delta}$ has zero elements, which is because we cannot catalytically distill a maximally coherent state from other state of the same dimension $r$, as is clear from the fact that $S_{\infty}(\phi^{\Delta})\le \log r$ which takes equality iff $\ket\phi$ is a $r$-dimensional maximally coherent state. Therefore, we can use Lemma~\ref{le:1} and only verify the case $\alpha>0$. 

As a result, we have
\begin{equation}\label{eq:56}
    C_{IO,pc}(\phi)=C_{IO,d}\left(\phi\right)=\log\floor{1/\phi^{\Delta}_{\max}}.
\end{equation}

\begin{remark}
Eq.~\eqref{eq:56} establishes that, although coherence catalysts are helpful for general pure state transformation, it is not in the task of coherence distillation (of pure state). Besides, the second equation of Eq.~\eqref{eq:56} is the same as the result in \cite{Regula-one-shot}, yet the proof is different.

\end{remark}


Now we deal with the case that allows a smoothing parameter $\varepsilon$ under the pure state smoothing restriction. We use the following well-known inequality between trace distance and fidelity \cite{fuchs99},
\begin{equation}\label{eq:fuchs}
    1-\mathrm{F}(\rho,\sigma)\le \mathrm{T}(\rho,\sigma) \le \mathrm{P}(\rho,\sigma),
\end{equation}
where $\mathrm{T}(\rho,\sigma) = \frac12 \|\rho-\sigma\|_1$ is the trace distance, and recall $\mathrm{P}(\rho,\sigma) = \sqrt{1-\mathrm{F}^2(\rho,\sigma)}$.

For pure states $\ket\phi$, $\ket\psi$, we have
\begin{equation}
    \mathrm{T}\left(\Delta\left(\ket\phi\bra\phi\right),\Delta\left(\ket\psi\bra\psi\right)\right) = \frac12\sum_{i=1}^n\left|\phi^{\Delta}_i-\psi^{\Delta}_i\right|,
\end{equation}which equals to the classical trace distance between vectors $\phi^\Delta$ and $\psi^\Delta$. We denote this as $\mathrm{T}\left(\phi^\Delta,\psi^\Delta\right)$ for simplicity. \\

Let the distillable coherence with perfect catalyst $\tilde{C}_{IO,pc}^\varepsilon(\phi)=N\le r$. Then, there exists a pure state $\tilde{\ket{{\Psi}}}$ such that $\phi^{\Delta} \prec_T \tilde{\Psi}^{\Delta}$ and $\mathrm{P}(\Psi_N,\tilde{{\Psi}})\le\varepsilon$. By Lemma~\ref{le:1}, this relation leads to
\begin{align}
    S_\infty(\phi^{\Delta})&\ge S_\infty(\tilde{\Psi}^{\Delta})\\
    \Leftrightarrow\quad  \phi^{\Delta}_{\max} &\le \tilde{\Psi}^{\Delta}_{\max}.\label{eq:phibound1}
\end{align}

Next, we bound the maximum element of vector $\tilde{\Psi}^\Delta$ and notice that
\begin{align}
    \mathrm{T}(\Psi_N^\Delta,\tilde{\Psi}^\Delta)&\le
    \mathrm{T}(\Psi_N,\tilde{\Psi})\\&\le
    \mathrm{P}(\Psi_N,\tilde{\Psi})\\&\le
    \varepsilon,
\end{align}
where the first inequality is due to the fact that trace distance is non-increasing under any quantum channel (here the dephasing channel $\Delta$), and the second one is due to Eq.~\eqref{eq:fuchs}. 

Based on the fact that the difference between the maximum elements of two normalized vectors is upper bounded by their trace distance, one gets
\begin{equation}
    \tilde{\Psi}^{\Delta}_{\max}\le\frac1N + \varepsilon,
\end{equation}
and as a consequence of  Eq.~\eqref{eq:phibound1}, we have
\begin{equation}
    \phi^{\Delta}_{\max} \le \frac1N + \varepsilon\label{eq:phibound}.
\end{equation}

Now we construct a pure state $\ket{\Psi^{*}}=\sum_{i=1}^{r}\sqrt{\Psi^{*\Delta}_i}\ket{i}$ with coefficient vector
\begin{equation}
    \Psi^{*\Delta}_i = 
    \begin{cases}
    \cfrac1N+\varepsilon ,& 1\le i \le N-1,\\
    \cfrac1N-\varepsilon(N-1) ,& i = N,\\
    0 ,& N+1\le i\le r.
    \end{cases}
\end{equation}
By our assumption $ r(r-1)\varepsilon \le 1$, hence $N(N-1)\varepsilon \le 1$, we see $\Psi^{*\Delta}$ is a non-negative and normalized vector. Together with Eq.~\eqref{eq:phibound}, we obtain
\begin{equation}
    \phi^{\Delta} \prec \Psi^{{*\Delta}},
\end{equation}
which makes the transformation $\ket\phi\xrightarrow{\text{IO,d}}\ket{\Psi^*}$ possible (without catalysts).

On the other hand, the fidelity between $\Psi_N$ and $\Psi^*$ can be bounded, 
 \begin{equation}
\begin{aligned}
    &\mathrm{F}\left(\Psi_N,\Psi^*\right)\\
    &= (N-1)\sqrt{\cfrac1N\left(\cfrac1N+\varepsilon\right)}+\sqrt{\cfrac1N\left(\cfrac1N-(N-1)\varepsilon\right)}\\
    &\ge (N-1)\sqrt{\cfrac1N\left(\cfrac1N-(N-1)\varepsilon\right)}+\sqrt{\cfrac1N\left(\cfrac1N-(N-1)\varepsilon\right)}\\
    &=\sqrt{1-N(N-1)\varepsilon}.
\end{aligned}
\end{equation}
This leads to the following bound of purified distance:
\begin{align}
    \mathrm{P}\left(\Psi_N,\Psi^*\right)\le \sqrt{N(N-1)\varepsilon} \le \sqrt{r(r-1)\varepsilon},
\end{align}
which establishes
\begin{equation}
    \tilde{C}_{IO,d}^{\sqrt{r(r-1)\varepsilon}}(\phi) \ge  \tilde{\mathrm{C}}^\varepsilon_{IO,pc}(\phi). 
\end{equation}
Together with the trivial inequality
\begin{equation}
    \tilde{C}_{IO,d}^{{\varepsilon}}(\phi) \le \tilde{\mathrm{C}}^\varepsilon_{IO,pc}(\phi),
\end{equation}
Thm.~\ref{th:perfect} is proved.




\section{Conclusion and Discussion}\label{sec:level4}
In this work, we introduce and systemically study the catalytic coherence distillation tasks. 
We show that, with unlimited resource as catalysts and global smoothing error, one can “distill” infinite amount of coherence with vanishingly small error from any initial state. This embezzling phenomenon is illustrated with two different protocols. We further define the distillable coherence with catalysts of restricted dimension, and derive a lower bound of this quantity. Moreover, we define and investigate the distillable coherence with perfect catalysts, by requiring the catalysts to be pure and precisely unchanged after the incoherent operations. In this setting, we find that the power of catalysts declines drastically. In particular,
for pure states and IO/SIO under a pure state smoothing restriction, catalyst provides almost no advantages for distillation. In general, our results enhance the understanding of catalytic effect in quantum resource theory.

Our work for the first time elucidates the existence of embezzling phenomenon in the resource theory of coherence. 
This phenomenon has been suggested in such as \cite{Streltsov-review, Eric18review}, but, to the best of our knowledge, we are the first to give explicit protocols to illustrate it in coherence resource theory. Also, we apply the recently proposed convex-split lemma to achieve one of our embezzling protocol, which is completely different from the previous embezzling state method in entanglement theory \cite{Embez}. 
Thus, our work also finds a new application of the convex-split lemma, which has been shown useful in other tasks such as in \cite{CataDec, CataRes}. Note that, \cite{EmbezThermo} discusses similar embezzling phenomenon in thermodynamics which is sometimes called ``unspeakable coherence'', but is quite different from the resource framework discussed in this paper.

What's more, our perfect catalysts scenario serves as a physically relevant definition of catalytic coherence distillation problem, settling the embezzling phenomenon and guaranteeing catalysts to be effective even after many rounds of distillation. However, in this scenario the operational power seems to decline. We hope that our different definitions can be inspiring for future studies on catalytic effect in quantum resource theory. 

Several questions are left open in our work. In the restricted catalysts scenario, it is practically interesting to explore a tight bound of distillable coherence and check whether there exists an embezzling protocol reaching the optimal rate. Similar questions about the embezzling efficiency is discussed in \cite{Embez}. In the perfect catalysts scenario, it is worth  exploring the catalytic advantage without the pure state smoothing restriction in Thm.~\ref{th:perfect}. Moreover,  similar studies can also be extended to mixed states and other free operation sets.\


\section*{ACKNOWLEDGEMENT}
This work is supported by the National Natural Science Foundation of China Grants No.~11674193 and No.~11875173, and the National Key R\&D Program of China Grants No.~2017YFA0303900 and No.~2017YFA0304004.
We thank Xiongfeng Ma, Yunchao Liu and Pei Zeng for the insightful discussions. 
\bibliographystyle{apsrev4-1}

\bibliography{CatalyticCoherenceDistillation}
\end{document}